\documentclass{article} 
\usepackage{authblk}
\title{Contiguous Allocation of Indivisible Items on a Path}
\author[1]{Yasushi Kawase}
\author[2]{Bodhayan Roy}
\author[2]{Mohammad Azharuddin Sanpui}
\affil[1]{The University of Tokyo}
\affil[2]{Indian Institute of Technology Kharagpur}
\date{}

\usepackage{fullpage}
\usepackage[linesnumbered,ruled,vlined]{algorithm2e}

\SetCommentSty{mycommfont}

\usepackage{amsthm,amsmath,amssymb,mathtools}
\newtheorem{theorem}{Theorem}

\newtheorem{proposition}{Proposition}
\newtheorem{example}{Example}

\newtheorem{corollary}{Corollary}

\usepackage{colortbl}
\newcolumntype{d}{>{\columncolor{gray!30}}c}
\usepackage{booktabs}

\usepackage[numbers,sort]{natbib} 
\usepackage{cleveref}
\usepackage{enumitem}
\setlist[itemize]{leftmargin=10pt}
\newcommand{\alert}[1]{\textcolor{red}{#1}}
\newcommand{\true}{\texttt{true}}
\newcommand{\false}{\texttt{false}}
\newcommand{\MMS}{\mathrm{MMS}}
\newcommand{\bA}{\mathbf{A}}
\renewcommand{\mid}{:}
\DeclareMathOperator*{\argmax}{arg\,max}
\newcommand{\ot}{\leftarrow}
\usepackage{tikz}
\usetikzlibrary{positioning,shapes.misc,calc}

\usepackage{todonotes}
\definecolor{myyellow}{cmyk}{0,0.02,0.23,0.01}

\newcommand*\divider[1]{\tikz[baseline=(T.base)]{\node[inner xsep=0pt, inner ysep=2pt, outer sep=0pt, fill=gray!30, overlay] {\phantom{$00$}};\node[inner xsep=0pt, inner ysep=1.5pt, outer sep=0pt] (T) {#1};}}

\begin{document}

\maketitle

\begin{abstract}
We study the problem of allocating indivisible items on a path among agents. 
The objective is to find a fair and efficient allocation in which each agent's bundle forms a contiguous block on the line. We say that an instance is \emph{$(a, b)$-sparse} if each agent values at most $a$ items positively and each item is valued positively by at most $b$ agents.
We demonstrate that, even when the valuations are binary additive, deciding whether every item can be allocated to an agent who wants it is NP-complete for the $(4,3)$-sparse instances.
Consequently, we provide two fixed-parameter tractable (FPT) algorithms for maximizing utilitarian social welfare, with respect to the number of agents and the number of items.
Additionally, we present a $2$-approximation algorithm for the special case when the valuations are binary additive, and the maximum utility is equal to the number of items. 
Also, we provide a $1/a$-approximation algorithm for the $(a,b)$-sparse instances.
Furthermore, we establish that deciding whether the maximum egalitarian social welfare is at least $2$ or at most $1$ is NP-complete for the $(6,3)$-sparse instances, even when the valuations are binary additive. We present a $1/a$-approximation algorithm for maximizing egalitarian social welfare for the $(a,b)$-sparse instances. Besides, we give two FPT algorithms for maximizing egalitarian social welfare in terms of the number of agents and the number of items.
We also explore the case where the order of the blocks of items allocated to the agents is predetermined.
In this case, we show that both maximum utilitarian social welfare and egalitarian social welfare can be computed in polynomial time.
However, we determine that checking the existence of an EF1 allocation is NP-complete, even when the valuations are binary additive.
\end{abstract}

\section{Introduction} \label{sec:intro}
Imagine a scenario in which multiple organizers wish to use the same conference center for their events. Each organizer has a preferred schedule for their events. Typically, organizers prefer to schedule their events in contiguous blocks of time rather than splitting them into separate periods. This leads to the following question: How should the conference center committee schedule time in a contiguous block of time for the different organizers?


A fundamental task in such allocation task is to achieve both fairness and efficiency.
Fair division is one of the most fundamental and well-studied topics in 
computational social choice theory~\cite{brandt2016introduction,chevaleyre2007short,aziz2016computational} and has received significant attention in the domains of mathematics, economics, political science, and computer science~\cite{Cake,Fairdivisionandcollectivewelfare,thomson2007children,brams1996fair,klamler2010fair}. 
Fair division problems are of particular interest because of their various real-world applications, such as students sharing the cost of renting an apartment, spouses sharing assets after divorce, and nations claiming ownership of disputed territories.
Research discussions on fair division often explore between two distinct categories of items. Certain items, such as cake and land, are considered divisible due to their ability to be divided among agents in an arbitrary manner~\cite{kurokawa2013cut,branzei2017query,aziz2014cake,edmonds2011cake,caragiannis2011towards,aziz2016discrete,aziz2020bounded}. Additional items, such as residences and automobiles, possess indivisible characteristics, necessitating their allocation in their whole to a single agent~\cite{segal2019democratic,bezakova2005allocating,bouveret2016fair,lipton2004approximately,lang2016fair,kawase2020max,aziz2023possible,halpern2020fair,Kawase2022RandomAO}.
This paper deals with indivisible items. For example, in the scheduling scenario, we consider the case where time slots (e.g., 10-minute increments) are provided in advance.

A natural criterion for assessing the quality of an allocation is \emph{utilitarian social welfare}, which is defined as the sum of the utilities among all agents. 
Another criterion is \emph{egalitarian social welfare}, which is defined as the minimum of the utilities of all agents.
One of the most prominent fairness notions is \emph{envy-freeness (EF)}, which means that no agent envies another based on the sets of items that they receive. 
Since EF is a strong fairness guarantee, there are also its relaxations to consider.
One standard such relaxation is \emph{envy-free up to one item (EF1)}, which requires that any envy that one agent has toward another can be eliminated by removing one item from the envied agent's bundle.
Other fairness criteria include \emph{maximin share guarantee (MMS)}, \emph{proportionality (PROP)}, and \emph{equitability (EQ)}. In a PROP allocation, each agent is guaranteed to receive at least a $1/n$ fraction of the value of the entire set of items, where $n$ is the total number of agents. In an EQ allocation, the values that correspond to the allocation of each agent must be equal.
The formal definition of these criteria will be provided in Section~\ref{sec:prelim}.

This paper explores the division of items that are arranged on a path while imposing the restriction that only contiguous subsets of items can be assigned to the agents. Our primary focus is on scenarios where each agent employs an additive valuation function, as this represents the most fundamental and crucial setting. We investigate the computational complexities of finding a contiguous allocation that meets a specified fairness and efficiency criterion.
Furthermore, we also examine a constraint where the blocks are assigned to agents in a specific order of agents.
In the scheduling scenario, this constraint means that the ordering of events is predetermined.

\subsection{Our results} \label{sec:results}
We investigate the computational complexities of allocating indivisible items on a path while ensuring contiguity.
We explore two settings: one in which the allocation must be consistent with a specified order of agents (fixed-order) and another in which it does not need to (flexible-order).


In Section~\ref{sec:fixord}, we examine the fixed-order setting.
We provide a polynomial time algorithm for maximizing utilitarian social welfare, based on dynamic programming (Section~\ref{sec:dp}).
We give a polynomial time greedy algorithm for maximizing egalitarian welfare (Section~\ref{sec:greedy}). 
Additionally, we present polynomial time algorithms for computing allocations that satisfy MMS, PROP, and EQ, respectively.
However, in Section~\ref{sec:NPEF1}, we demonstrate that deciding the existence of an EF1 allocation is NP-hard, even when the valuations are binary additive. 

In Section~\ref{sec:flexord}, we consider the flexible-order setting.
We prove that both maximizing utilitarian social welfare and maximizing egalitarian social welfare are NP-hard in sparse instances of the flexible-order setting, in contrast to the fixed-order setting.
This hardness holds even when the valuations are binary additive, and the question is to determine whether the optimal utilitarian social welfare is equal to the number of items.
Moreover, it is NP-complete even when the valuations are binary additive, and the question is to decide whether the maximum egalitarian social welfare is at least $2$ or at most $1$.
Consequently, we provide a $2$-approximation algorithm for maximizing utilitarian social welfare, when the valuations are binary additive and the optimum utilitarian social welfare is the number of items.
We also present two $a$-approximation algorithms for the $(a,b)$-sparse instances for maximizing utilitarian and egalitarian social welfare, respectively.
Furthermore, we present two FPT algorithms for maximizing utilitarian social welfare in terms of the number of agents and the number of items. We also provide two FPT algorithms for maximizing egalitarian social welfare in terms of   the number of agents and the number of items. 
For envy-freeness, it is known that an EF1 allocation always exists in the flexible-order setting~\cite{igarashi2023cut}.
However, the complexity of constructing it remains an open question.


Our results are summarized in Table~\ref{tbl:summary}.
For the sake of comparison, we also provide the results for the case without contiguity constraint.

\begin{table}[htb]
\caption{The computational complexities of checking the existence of an allocation of items that satisfies a designated property and constructing one, if it exists. All the hardness results hold even when the valuations are restricted to be binary additive.}\label{tbl:summary}
\centering
\begin{tabular}{c|lll}
\toprule
       & \multicolumn{2}{c}{contiguous} & unconstrained\\
       & flexible-order      & fixed-order& \\\midrule
EF               & NP-h~\cite{goldberg2020contiguous} & NP-h~\cite{goldberg2020contiguous}       & NP-h~\cite{aziz2015fair}  \\
EF1              & open      & \textbf{NP-h {\small(Theorem~\ref{thm:fixed-EF1-hard})}}                   & P~\cite{lipton2004approximately,caragiannis2019unreasonable} \\
U-max            & \textbf{NP-h~{\small(Theorem~\ref{thm:Umax-hard})}}         & \textbf{P {\small(Theorem~\ref{thm:maxsum})}} &P${}^\dagger$            \\
E-max            & \textbf{NP-h~{\small(Theorem~\ref{thm:Emax-hard})}}         & \textbf{P {\small(Theorem~\ref{thm:Emax})}} &P${}^\ddagger$              \\
PO               & P~\cite{igarashi2019pareto}  & \textbf{P~{\small(Theorem~\ref{thm:maxsum})}} &P${}^\dagger$              \\
MMS              & P~\cite{bouveret2017fair}       & \textbf{P~{\small(Theorem~\ref{thm:mms})}} & P~\cite{bouveret2016characterizing} \\
PROP             & NP-h~\cite{goldberg2020contiguous} & \textbf{P {\small(Theorem~\ref{thm:prop})}}              & P${}^\ddagger$\\
EQ               & NP-h~\cite{goldberg2020contiguous} & \textbf{P {\small(Theorem~\ref{thm:equit})}}& P${}^\ddagger$\\
\bottomrule
\end{tabular}
\begin{flushleft}
{\scriptsize
${}^\dagger$ These can be solved by just allocating each item to an arbitrary agent who values it.\\
${}^\ddagger$ These can be solved by a max-flow algorithm (see, e.g., \cite{Schrijver2003}).\\
}
\end{flushleft}
\end{table}

\subsection{Related work} \label{sec:related}
The contiguity requirement has been studied in connection to fairness notions in the context of allocating divisible items, often signified by a cake. In particular, Even and Paz~\cite{even1984note} showed the existence of contiguous proportional allocation using the divide and conquer rule. Dubins and Spanier~\cite{dubins1961cut} presented a moving-knife technique that ensures a contiguous proportional allocation. Stromquist \cite{stromquist1980cut} stated a moving-knife algorithm that provides a guarantee of a contiguous envy-free allocation for a group of three players and also confirmed the existence of a contiguous envy-free allocation, but unachievable using a finite algorithm \cite{stromquist2008envy}. 
Su~\cite{edward1999rental} used approaches that included Sperner's lemma in order to show the existence of a contiguous envy-free allocation. Dobo{\v{s}} et al.~\cite{dobovs2013existence} explored the concept of contiguous equitable allocations, focusing on their existence and computation. Notably, they demonstrated that the presence of such an allocation is certain, even when the ordering of the agents is predetermined.
Arunachaleswaran et al.~\cite{arunachaleswaran2019fair} gave an algorithm that efficiently computes a contiguous cake division with envy that is multiplicatively restricted. Specifically, the envy noticed by each agent is limited by a factor of $3$. Barman and Kulkarni~\cite{barman2022approximation} designed a computationally efficient algorithm that yields a contiguous cake division with both additive and multiplicative restrictions on envy. An algorithm proposed by Deng et al.~\cite{deng2012algorithmic} provides an additive approximation of an envy-free connected piece cake division. However, this process requires an exponential amount of time in terms of the number of agents. Goldberg et al.~\cite{goldberg2020contiguous} provided an efficient algorithm for the computation of contiguous allocations, ensuring that the envy between any two agents does not exceed one-third. 

In the case of indivisible items, the contiguity requirement has also been taken into account. Marenco and Tetzlaff~\cite{marenco2014envy} showed that under the condition that items are placed on a line and each item is valued positively by no more than one agent, the existence of a contiguous envy-free allocation is guaranteed. Barrera et al.~\cite{barrera2015discrete}, Bil{\`o} et al.~\cite{bilo2022almost}, and Suksompong~\cite{suksompong2019fairly} proved that different relaxations of envy-freeness can be satisfied when each item has an incentive to provide positive value for any number of agents. Bouveret et al.~\cite{bouveret2017fair} showed that the problem of detecting the existence of a contiguous fair allocation is NP-hard while considering either proportionality or envy-freeness as fairness criteria. Goldberg et al.~\cite{goldberg2020contiguous} showed that it is NP-hard to decide if an instance with indivisible items on a line admits a contiguous allocation satisfying all properties in $X$, even if all agents have binary valuations and value the same number of items, where $\emptyset\ne X\subseteq\{\textit{EF, PROP, EQ}\}$. 
Contiguity has been the subject of research in the broader context of indivisible items placed on a graph of arbitrary structure \cite{igarashi2019pareto,Bei2019ConnectedFA,igarashi2023cut}. The notion of graph fair division has received significant attention in subsequent years~\cite{aziz2018knowledge,barany2015block,truszczynski2020maximin,greco2020complexity,igarashi2023cut,chevaleyre2017distributed,Amanatidis2022FairDO,aziz2022algorithmic,misra2022fair}.

In recent times, there has been an increasing focus on the issue of social welfare in the context of cake-cutting. The studies of this specific topic started by Caragiannis et al.~\cite{caragiannis2012efficiency}, with the objective of quantifying the demise in social welfare that may potentially arise from various fairness criteria. 
Aumann et al.~\cite{aumann2012computing} investigated the problem of finding a contiguous allocation that maximizes utilitarian social welfare for both divisible and indivisible items. 
They proved that finding optimal contiguous allocation is NP-hard, using a reduction from the $3$-dimensional matching problem, even when the valuations are piecewise-uniform. They also provide a $1/8$-approximation algorithm when valuations are additive. Bei et al.~\cite{bei2012optimal} and Cohler et al.~\cite{cohler2011optimal} considered maximizing utilitarian social welfare but introduced additional fairness constraints of proportionality and envy-freeness, respectively. Bei et al.~\cite{bei2012optimal} also developed approximation results for maximizing utilitarian social welfare with proportionality as a constraint. In the context of indivisible items on a path, Misra et al.~\cite{misra2021equitable} showed that maximizing utilitarian social welfare is NP-hard even on binary $(4,4)$-sparse instances. We will prove the NP-hardness for the class of binary-additive on $(4,3)$-sparse instances using a reduction from 2L-OCC-3SAT. Igarashi and Peters~\cite{igarashi2019pareto} provided a polynomial-time algorithm that finds a Pareto optimal contiguous allocation when valuations are additive.
   

\section{Preliminaries} \label{sec:prelim}
We study the problem of allocating indivisible items where the items are arranged on a line, and each agent has an additive valuation for each item.
We call this setting \emph{contiguous allocation of items on a path (CAP)}.
For a positive integer $k$, we denote the set $\{1,2,\dots,k\}$ by $[k]$.
Let ${M}=\{g_1,g_2,\dots,g_m\}$ denote the set of $m$ indivisible items, and ${N}=[n]$ be the set of agents. 
Assume that items are aligned on a path in the order of indices.
Each agent $i\in {N}$ has an additive valuation $v_i\colon 2^M\to\mathbb{Z}_+$ where $v_i(X)=\sum_{g\in X} v_i(\{g\})$ for each $X \subseteq {M}$.
We write $v_i(g)$ to denote $v_i(\{g\})$ for short.
The valuation $v_i$ is called \emph{binary} if $v_i(g)\in\{0,1\}$ for each $g\in M$.
For simplicity, we assume that each item $g\in {M}$ is valued, i.e., there exists an agent $i\in {N}$ such that $v_i(g)\ge 1$. 
An instance of CAP is $(N,M,(v_i)_{i\in N})$.
We call an instance $(N,M,(v_i)_{i\in N})$ of CAP \emph{binary} if $v_i$ is binary for every $i\in N$.

An allocation $\bA=(A_1,A_2,\dots,A_n)$ is a partition of all items into bundles for the agents, i.e., $\bigcup_{i\in N}A_i=M$ and $A_i\cap A_j=\emptyset$ for any distinct $i,j\in N$.
In the allocation $\bA$, agent $i\in N$ receives bundle $A_i$. 
We call an allocation $\bA$ \emph{contiguous} if each bundle $A_i$ forms a contiguous block of items on the line, i.e., $A_i=\{g_k,g_{k+1},\dots,g_{\ell}\}$ for some $k$ and $\ell$.
Moreover, a contiguous allocation $\bA$ is called \emph{order-consistent} if the blocks are assigned to agents in a specific order, i.e., there exist indices $1=k_1\le k_2\le\dots\le k_n\le k_{n+1}=m+1$ such that $A_i=\{g_{k_i},g_{k_i+1},\dots,g_{k_{i+1}-1}\}$ for each $i\in N$. 
We consider two settings: the fixed-order setting and the flexible-order setting.
In the fixed-order setting, we only allow contiguous allocations that are order-consistent.
In the flexible-order setting, we allow all the contiguous allocations.

An allocation $\bA$ is called \emph{envy-free} if, for all $i,j\in N$, it holds that $v_i(A_i)\ge v_i(A_j)$.
In addition, an allocation $\bA$ is called \emph{envy-free up to one item (EF1)} if, for all $i,j\in N$, it holds that $v_i(A_i)\ge v_i(A_j\setminus X)$ for some $X\subseteq A_j$ with $|X|\le 1$.
The \emph{utilitarian social welfare} and the \emph{egalitarian social welfare} of an allocation $\bA$ are defined as $\sum_{i\in N}v_i(A_i)$ and $\min_{i\in N}v_i(A_i)$, respectively.
We call an allocation $\bA$ is \emph{U-max} and \emph{E-max} if it maximizes the utilitarian social welfare and the egalitarian social welfare, respectively.
An allocation $\bA$ is called \emph{Pareto-optimal (PO)} if, for any other allocation $\bA'$, we have $v_i(A_i)=v_i(A_i')$ for all $i\in N$ or $v_i(A_i)>v_i(A_i')$ for some $i\in N$.
Clearly, any U-max allocation is PO.
The maximin share guarantee of an agent $i$ is defined as $\MMS(i)=\max_{\bA\in \mathcal{A}} {\min_{j\in N}}v_i(A_j)$, where $\mathcal{A}$ is the set of all possible contiguous allocations.
An allocation $\bA$ is said to be \emph{maximin share (MMS)} if $v_i(A_i)\geq \MMS(i)$ for every $i\in N$. 
Moreover, an allocation $\bA$ is said to be \emph{proportional (PROP)} and \emph{equitable (EQ)} if $v_i(A_i)\ge v_i(M)/n~(\forall i\in N)$ and $v_i(A_i)=v_j(A_j)~(\forall i,j\in N)$, respectively.

We say that an instance is \emph{$(a, b)$-sparse} if each agent values at most $a$ items positively and each item is valued positively by at most $b$ agents. When the valuations are binary, then the $(a,b)$-sparse instance is called binary $(a,b)$-sparse instance.
\section{Fixed Order Setting} \label{sec:fixord}
In this section, we explore the setting where the allocation must be order-consistent.
To clarify our setting, let's begin by observing a specific example.
\begin{example}\label{ex:fixed-order}
Suppose that there are two agents $N=\{1,2\}$ and four items $M=\{g_1,g_2,g_3,g_4\}$.
The agents' valuations are given as 
$v_1(g_1)=v_1(g_2)=v_1(g_3)=v_1(g_4)=v_2(g_1)=v_2(g_2)=1$ and 
$v_2(g_3)=v_2(g_4)=0$ (see Figure~\ref{fig:ex-fixed-order}).
Note that there are five possible order-consistent contiguous allocations.

No order-consistent contiguous allocation is EF1 (and also EF) because $v_1(A_1)+2\le v_1(A_2)$ if $A_1\subseteq\{g_1\}$ and $v_2(A_2)+2\le v_2(A_1)$ if $A_1\supsetneq\{g_1\}$.
The allocation that agent $1$ receives all the items is U-max, and the resulting utilitarian social welfare is $4$.
The allocation $(A_1,A_2)=(\{g_1\},\{g_2,g_3,g_4\})$ is E-max and EQ, where $v_1(A_1)=v_2(A_2)=1$.
In addition, there are no MMS allocations or PROP allocations since $\MMS(1)=v_1(M)/2=2$ and $\MMS(2)=v_2(M)/2=1$.
\begin{figure}[htb]
\centering
\begin{tikzpicture}[scale=1.3]
  \draw (1,0) -- (4,0);
  \foreach \i/\s/\t in {1/1/1, 2/1/1, 3/1/0, 4/1/0} 
  {
    \draw[fill=white] (\i,0) circle (1.8mm) node {$g_\i$};
    \node at (\i,-.35) {$\s$};
    \node at (\i,-.6) {$\t$};
  }
  \node at (0.2,-.36) {$v_1$:};
  \node at (0.2,-.61) {$v_2$:};
\end{tikzpicture}
\caption{The valuation of the agents in Example~\ref{ex:fixed-order}} \label{fig:ex-fixed-order}
\end{figure}
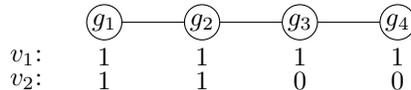
\end{example}

In what follows, we examine the problem of checking the existence of an allocation that satisfies a designated property and constructing one if it exists. 
We provide polynomial-time algorithms for the following properties: U-max, E-max, EQ, PROP, MMS.
Moreover, we demonstrate that the problem of verifying the existence of an EF1 allocation is NP-hard.

\subsection{Dynamic programming for U-max and EQ} \label{sec:dp}
We first design an algorithm for the U-max problem based on dynamic programming.
Our algorithm utilizes a table, denoted as $T$, with rows $[n]$ and columns $\{0,1,\dots,m\}$. Here, $n$ represents the number of agents, and $m$ represents the number of items. 
Each cell $T_{i,j}$ in the table represents the maximum utilitarian social welfare when allocating the first $j$ items (i.e., $\{g_1,g_2,\dots,g_j\}$) to the first $i$ agents (i.e., $\{1,2,\dots,i\}$). Note that some of these agents may not receive any item. We will show how to compute the allocation for each cell of the table in polynomial time.

For every $j\in\{0,1,\dots,m\}$, the entry $T_{1,j}$ is $v_1(\{g_1,\dots,g_j\})$. Hence, $T_{1,j}$ can be computed in a constant time per entry.
In addition, the entry $T_{i,0}$ is $0$ for every $i\in\{1,2,\dots,n\}$.
For each $i\in\{2,3,\dots,n\}$ and $j\in\{1,2,\dots,m\}$, the entry $T_{i,j}$ can be determined as
\begin{align*}
T_{i,j}=\max\{T_{i,j-1}+v_i(g_j),\,T_{i-1,j}\}
\end{align*}
by considering two cases of whether item $g_j$ is allocated to agent $i$ or not.
Thus, each entry can be computed in a constant time by filling them in an appropriate order.
Hence, the objective value $T_{n,m}$ can be computed in $O(mn)$ time.
Moreover, an allocation that attains the value can be constructed in linear time by backtracking the table. Our algorithm is formally described as Algorithm~\ref{alg:U-max}. 

\begin{algorithm}[H]
    \caption{Dynamic Programming for U-max}\label{alg:U-max}
    \tcc{Construct a table $T$}
    \lFor{$i\gets1$ \KwTo $n$}{$T_{i,0}\ot 0$}
    \lFor{$j\gets1$ \KwTo $m$}{$T_{1,j}\ot T_{i,j-1}+v_1(g_j)$}
    \For{$i\gets2$  \KwTo $n$} {
        \For{$j\gets1$ \KwTo   $m$}{
            $T_{i,j}=\max\{T_{i,{j-1}}+v_i(g_j),\,T_{{i-1},j}\}$\;
        }
    }
    \tcc{Output an allocation $\bA^*$ with the U-max value $T_{m,n}$}
    Let $\bA^*\ot (\emptyset,\emptyset,\dots,\emptyset)$, $i\ot n$, $j\ot m$, and $u^*\ot T_{n,m}$\;
    \While{$i\ge 2$ and $j\ge 1$}{
      \If{$T_{i,j-1}+v_i(g_j)=u^*$}{
        Update $A^*_i\ot A^*_i\cup\{g_j\}$, $u^*\ot u^*-v_i(g_j)$, and $j\ot j-1$\;
      }
      \Else{
        $i\ot i-1$\;
      }
    }
    \lIf{$i=1$}{$A^*_1\ot \{g_1,\dots,g_j\}$}
    \Return the allocation $\bA^*$\;
\end{algorithm}

Therefore, we obtain the following theorem.
\begin{theorem}\label{thm:maxsum}
For the fixed-order setting, the U-max problem of CAP can be solved in $O(mn)$ time.
\end{theorem}
Note that even when the valuation functions are not necessarily additive, the U-max problem can be solved in $O(m^2n)$ time by updating the table as $T_{i,j}=\max_{\ell=1}^{j+1}(T_{i-1,\ell-1}+v_i(\{g_\ell,g_{\ell+1},\dots,g_{j}\}))$.
Moreover, the algorithm implies an FPT algorithm with respect to the number of agents for the flexible-order setting.
\begin{corollary}\label{cor:Umax}
For the flexible-order setting, 
the U-max problem of CAP is FPT with respect to the number of agents.
\end{corollary}
\begin{proof}
We may apply the dynamic programming algorithm for all possible orderings of agents. The total number of such orderings is $n!~(\le n^n)$. 
Therefore, for the flexible-order setting, the U-max problem can be solved in $n!\cdot O(mn)=O(n^{n+1}\cdot m)$ time, which is FPT with respect to the number of agents.
\end{proof}
It should be noted that \citet{aumann2012computing} provided a faster FPT algorithm for this setting by directly applying dynamic programming. Their algorithm is based on a technique similar to Held--Karp algorithm for the traveling salesman problem.

Next, we provide an algorithm for the EQ problem. 
To do this, we solve the problem of deciding whether there exists an order-consistent contiguous allocation $\bA$ such that $v_i(A_i)=\alpha~(\forall i\in N)$ for each 
\begin{align*}
\alpha\in\big\{v_1(\emptyset),v_1(\{g_1\}),v_1(\{g_1,g_2\}),\dots,v_1(\{g_1,g_2,\dots,g_m\})\big\}.
\end{align*}
Hereafter, we solve the decision problem for a fixed $\alpha$.
To solve the problem, we construct a table $P$ with $n$ rows and $m$ columns.
For each $i\in [n]$ and $j\in[m]$,
each cell $P_{i,j}$ represents the existence of an order-consistent contiguous allocation of items $\{g_1,\dots,g_j\}$ to agents $\{1,2,\dots,i\}$ such that each agent values received block as $\alpha$.
Then, $P_{0,0}$ is $\true$, 
$P_{i,0}$ is $\true$ for every $i\in\{1,\dots,n\}$ if $\alpha=0$,
$P_{i,0}$ is $\false$ for every $i\in\{1,\dots,n\}$ if $\alpha\ne 0$,
and $P_{0,j}$ is $\false$ for every $j\in\{1,2,\dots,m\}$.
Moreover, for each $i\in[n]$ and $j\in[m]$, the entry $P_{i,j}$ is $\true$ if and only if there exists an index $\ell\in[j+1]$ such that $P_{i-1,\ell-1}$ is $\true$ and $v_i(\{g_\ell,g_{\ell+1},\dots,g_j\})=\alpha$. Thus, each entry can be computed in $O(m)$ time, and all the entries can be filled in $O(nm^2)$ time.
The desired allocation exists if $P_{n,m}=\true$, and such an allocation can be constructed in linear time by backtracking the table. Our algorithm is formally described as Algorithm~\ref{alg:EQ}.

\begin{algorithm}[H]
    \caption{Dynamic Programming for EQ}\label{alg:EQ}
    \For{$p\gets 0$ \KwTo $m$}{
        Let $\alpha \ot \sum_{k=1}^pv_i(g_k)$\;
        \tcc{Construct a table with respect to $\alpha$}
        
        $P_{0,0}\ot \true$\;
        \For{$i\gets 1$ \KwTo $n$}{
            \lIf{$\alpha=0$}{$P_{i,0}\ot \true$}
            \lElse{$P_{i,0}\ot \false$}
        }
        
        \lFor{$j\gets 1$ \KwTo $m$}{$P_{0,j}\ot \false$}
        \For{$i\gets 1$ \KwTo $n$}{
            \For{$j\gets 1$ \KwTo $m$}{
                $P_{i,j}\ot \false$\;
                \For{$\ell\gets j+1$ \KwTo $1$}{
                    \lIf{$P_{i-1,\ell-1}=\true$ and $\sum_{k=\ell}^j v_i(g_k)=\alpha$}{$P_{i,j}\ot \true$}
                }
            }
        }
        \If{$P_{n,m}=\true$}{
            \tcc{Output an allocation $\bA$ such that $v_i(A_i)=\alpha$ for all $i\in N$}
            $j\ot m$\;
            \For{$i\gets n$ \KwTo $1$}{
                Let $\ell$ be an index in $[j+1]$ such that $P_{i-1,\ell-1}=\true$ and $\sum_{k=\ell}^jv_i(g_k)=\alpha$\;
                Let $A_i\ot \{g_\ell,g_{\ell+1},\dots,g_j\}$ and $j\ot \ell-1$\;
                \If{$j=0$}{
                    $A_t\ot\emptyset$ for each $t=1,2,\dots,i-1$ and \textbf{break}\;
                }
            }
            \Return $(A_1,\dots,A_n)$\;
        }
     }
    \Return ``No equitable allocation exists''\;
\end{algorithm}

This algorithm implies the following theorem.
\begin{theorem}\label{thm:equit}
For the fixed-order setting, the EQ problem of CAP can be solved in $O(nm^3)$ time.
\end{theorem}
We remark that Algorithm~\ref{alg:EQ} can be easily extended to the case where the valuations are not restricted to additive.

\subsection{Algorithms for E-max, PROP, and MMS} \label{sec:greedy}
In this subsection, we provide polynomial-time algorithms to solve E-max, PROP, and MMS.

We first consider the E-max problem.
Here, we assume that $m\ge n$ since otherwise the E-max value must be $0$.
The E-max value must be in 
\begin{align*}
S=\{0\}\cup \{v_i(\{g_j,g_{j+1},\dots,g_\ell\})\mid i\in N,~j,\ell\in[m],~j<\ell\}.
\end{align*}
Note that $|S|=O(nm^2)$.
Hence, it is sufficient to solve the problem of deciding whether there exists an order-consistent contiguous allocation $\bA$ such that $\min_{i\in N}v_i(A_i)\ge \alpha$ for a given $\alpha\in S$. 
If such an allocation exists, the E-max value is at least $\alpha$.
Let $k_0=0$ and $k_i=\min\{k\mid v_i(g_{k_{i-1}+1})+\dots+v_i(g_k)\ge \alpha\}$ for each $i\in[n]$.
Then, $k_1,\dots,k_n\in\{0,1,\dots,m\}$ if and only if a desired allocation exists for the decision problem.
Indeed, if such indices exist, the egalitarian social welfare is at least $\alpha$ for the allocation $\bA$ such that $A_i=\{g_{k_{i-1}+1},\dots,g_{k_i}\}$ for each $i\in\{1,2,\dots,n-1\}$ and $A_n=\{g_{k_{n-1}+1},\dots,g_{m}\}$.
Thus, the decision problem can be solved in linear time.
The algorithm is described as Algorithm~\ref{alg:E-max}

\begin{algorithm}[H]
    \caption{Polynomial-time Algorithm to decide E-max is at least $\alpha$}\label{alg:E-max}
    Let $k_0\ot 0$\;
    \For{$i\gets 1$ \KwTo $n-1$}{
        $u\ot 0$\;
        \For{$k\gets k_{i-1}$ \KwTo $m$}{
            \If{$u\ge \alpha$}{
                $k_i\ot k$\;
                $A_i\ot \{g_{k_{i-1}+1},\dots,g_{k_i}\}$\;
                \textbf{break}\;
            }
            \If{$k=m$}{
                \Return ``No contiguous allocation has egalitarian social welfare at least $\alpha$''\;
            }
            $u\ot u+v_i(g_{k+1})$\;
            
        }
    }
    Let $A_n\ot\{g_{k_{n-1}+1},\dots,g_m\}$\;
    \If{$v_n(A_n)<\alpha$}{
        \Return ``No contiguous allocation has egalitarian social welfare at least $\alpha$''\;
    }
    \Return $(A_1,\dots,A_n)$\;
\end{algorithm}

The minimum $\alpha$ for which the answer to the decision problem is true can be found by the binary search. 
Thus, the decision problem is solved at most $O(\log |S|)=O(\log m)$ times.
Additionally, it takes $O(nm^2\log m)$ time to sort the elements of $S$.
Therefore, we get the following theorem.
\begin{theorem}\label{thm:Emax}
For the fixed-order setting, the E-max problem of CAP can be solved in $O(nm^2\log m)$ time.
\end{theorem}

This algorithm implies an FPT algorithm with respect to the number of agents for the flexible-order setting, similar to the discussion for Corollary~\ref{cor:Umax}, by considering all possible orderings of agents.
\begin{corollary}\label{cor:Emax}
For the flexible-order setting, 
the E-max problem of CAP is FPT with respect to the number of agents.
\end{corollary}

A PROP allocation can be constructed by setting $k_0=0$, $k_i=\min\{k\mid v_i(g_{k_{i-1}+1})+\dots+v_i(g_k)\ge v_i(M)/n\}$ for each $i\in [n-1]$, and $k_n=m$, if exists.
As $v_i(M)$ can be computed in $O(m)$ time for each $i\in N$, we get the following theorem.
\begin{theorem}\label{thm:prop}
For the fixed-order setting, the PROP problem of CAP can be solved in $O(nm)$ time.
\end{theorem}
            

Moreover, an MMS allocation can be constructed by setting $k_0=0$, $k_i=\min\{k\mid v_i(g_{k_{i-1}+1})+\dots+v_i(g_k)\ge \MMS(i)\}$ for each $i\in[n-1]$, and $k_n=m$, if exists.
Here, the value of $\MMS(i)$ for each $i\in N$ can be computed by solving the E-max problem in a situation where there are $n$ agents with valuation function $v_i$.
The E-max problem can be solved in $O(nm^2\log m)$ time by Theorem~\ref{thm:Emax}. 
Therefore, we obtain the following theorem.
\begin{theorem}\label{thm:mms}
For the fixed-order setting, the MMS problem of CAP can be solved in $O(n^2m^2\log m)$ time.
\end{theorem}

It is worth mentioning that the problems of E-max, PROP, and MMS can be solved in polynomial time, even for general valuation functions, by utilizing a table similar to the one used in Algorithm~\ref{alg:EQ}.

\subsection{NP-hardness of EF1} \label{sec:NPEF1}
Here, we demonstrate that the problem of checking the existence of an EF1 allocation for CBP in the fixed-order setting is NP-complete.

\begin{theorem}\label{thm:fixed-EF1-hard}
For the fixed-order setting, the EF1 problem of CBP is NP-hard.
\end{theorem}
\begin{proof}
We provide a reduction from 3SAT, which is known to be NP-complete~\cite{GJ1979}.
Let $\varphi$ be a 3SAT formula, where $\varphi=C_1\wedge C_2\wedge\dots\wedge C_m$ and each $C_j$ is a clause of the form $C_j=\ell_{j,1}\vee\ell_{j,2}\vee\ell_{j,3}$. In each clause, every literal $\ell_{j,1}, \ell_{j,2}, \ell_{j,3}$ is one of the variables $x_1,\dots,x_n$ or their negations $\overline{x}_1,\dots,\overline{x}_n$.
Without loss of generality, we assume that $m\ge 1$.

We create $n$ positive variable agents $N_x=\{x_1,\dots,x_n\}$, $n$ negative variable agents $N_{\overline{x}}=\{\overline{x}_1,\dots,\overline{x}_n\}$, and $3m$ clause agents $N_c=\bigcup_{j=1}^m \{c_{j,1},c_{j,2},c_{j,3}\}$.
Additionally, we create $n+m$ divider agents $N_d=\{d_1,\dots,d_{n+m}\}$ and $2$ stopper agents $N_s=\{s_1,s_2\}$. 
The set of agents is $N=N_x\cup N_c\cup N_d\cup N_s$, which consists of $r\coloneqq 3n+4m+2$ agents.
Suppose that the agents are ordered as follows:
\begin{align*}
\begin{split}
&s_1, s_2, d_1, x_1, \overline{x}_1, d_2, x_2, \overline{x}_2, d_3, \ldots, 
d_n, x_n, \overline{x}_n, d_{n+1},\\ 
&\quad c_{1,1}, c_{1,2}, c_{1,3}, d_{n+2},
c_{2,1}, c_{2,2}, c_{2,3}, d_{n+3},
\ldots, d_{n+m}, c_{m,1}, c_{m,2}, c_{m,3}.
\end{split}
\end{align*}

We create $n$ sets of variable items $V_1,\dots,V_n$,
$m$ sets of clause items $Q_1,\dots,Q_m$, and
$n+m$ sets of divider items $D_1,\dots,D_{(n+m)}$.
Each set of variable, clause, and divider items consists of 
$3$, $7$, and $r+2$ indivisible items, respectively.
The set of items is $M=\bigcup_{i=1}^n V_i\cup\bigcup_{j=1}^m Q_j\cup \bigcup_{k=1}^{n+m}D_k$, which contains $3n+7m+(r+2)(n+m)$ indivisible items.
The items are ordered according to the following:
\begin{align*}
&D_1, V_1, D_2, V_2,\dots, D_n,V_n,D_{n+1},Q_1,D_{n+2},\dots,D_{n+m},Q_m.
\end{align*}

For each set of variable items $V_i$ with $i\in[n]$, 
$x_i$ wants the left two items, $\overline{x}_i$ wants the right two items, and all the other agents do not want any items.
For each set of clause items $Q_j$ with $j\in[m]$, 
$\ell_{j,1}$ wants the three leftmost items,
$\ell_{j,2}$ wants the middle three ($3$rd, $4$th, and $5$th) items,
$\ell_{j,3}$ wants the three rightmost items,
$c_{j,1},c_{j,2},c_{j,3}$ want all the seven items,
and the other agents do not want any items.
For each divider set $D_i$ with $i\in[n+m]$, 
$s_1$ values the right two items, $s_2$ values the left two items if $i\ge 2$, 
$d_i$ values all the $m+2$ items, and the other agents do not value any items.

Intuitively, the truth assignment of $\varphi$ corresponds to whether the middle item of $V_i$ is allocated to $x_i$ or $\overline{x}_i$.
Whether the truth assignment satisfies the $j$th clause $C_j$ corresponds to whether $c_{j,1}, c_{j,2}, c_{j,3}$ are not envied by more than one item.

For example, if $\varphi=C_1\wedge C_2$ where $C_1=x_1\vee x_2\vee \overline{x}_3$ and $C_2=x_2\vee x_3\vee \overline{x}_4$, the valuations of the agents in the reduced instance is given as in \Cref{tab:fixed-EF1-hard}.

\begin{table*}[htbp]
\centering
\caption{Reduced instance of the EF1 problem of CBP from $(x_1\wedge x_2\wedge\overline{x}_3)\wedge(x_2\wedge x_3\wedge \overline{x}_4)$. The allocation represented by red color corresponds to a truth assignment of $(x_1,x_2,x_3,x_4)=(\true,\false,\false,\false)$.}\label{tab:fixed-EF1-hard}
\tabcolsep = 1.5mm
\renewcommand{\arraystretch}{.5}
\scalebox{1}{
\begin{tabular}{c||d|c|d|c|d|c|d|c|d|c|d|c}
\toprule
            &$D_1$   &$V_1$&$D_2$   &$V_2$&$D_3$   &$V_3$&$D_4$   &$V_4$&$D_5$   &$Q_1$&$D_6$   &$Q_2$\\\midrule
$s_1$       &0$\cdots$011& 000     &0$\cdots$011& 000     &0$\cdots$011& 000     &0$\cdots$011& 000    &0$\cdots$011& 0000000 &0$\cdots$011& 0000000\\
$s_2$       &00$\cdots$00& 000     &110$\cdots$0& 000     &110$\cdots$0& 000     &110$\cdots$0& 000    &110$\cdots$0& 0000000 &110$\cdots$0& 0000000\\\hline
$d_1$       &\alert{11$\cdots$1}1& 000     &00$\cdots$00& 000     &00$\cdots$00& 000     &00$\cdots$00& 000    &00$\cdots$00& 0000000 &00$\cdots$00& 0000000\\
$x_1$       &00$\cdots$0\alert{0}& \alert{11}0     &00$\cdots$00& 000     &00$\cdots$00& 000     &00$\cdots$00& 000    &00$\cdots$00& 1110000 &00$\cdots$00& 0000000\\
$\overline{x}_1$ &00$\cdots$00& 01\alert{1}     &\alert{0}0$\cdots$00& 000     &00$\cdots$00& 000     &00$\cdots$00& 000    &00$\cdots$00& 0000000 &00$\cdots$00& 0000000\\
$d_2$       &00$\cdots$00& 000     &1\alert{1$\cdots$1}1& 000     &00$\cdots$00& 000     &00$\cdots$00& 000    &00$\cdots$00& 0000000 &00$\cdots$00& 0000000\\
$x_2$       &00$\cdots$00& 000     &00$\cdots$0\alert{0}& \alert{1}10     &00$\cdots$00& 000     &00$\cdots$00& 000    &00$\cdots$00& 0011100 &00$\cdots$00& 1110000\\
$\overline{x}_2$ &00$\cdots$00& 000     &00$\cdots$00& 0\alert{11}     &\alert{0}0$\cdots$00& 000     &00$\cdots$00& 000    &00$\cdots$00& 0000000 &00$\cdots$00& 0000000\\
$d_3$       &00$\cdots$00& 000     &00$\cdots$00& 000     &1\alert{1$\cdots$1}1& 000     &00$\cdots$00& 000    &00$\cdots$00& 0000000 &00$\cdots$00& 0000000\\
$x_3$       &00$\cdots$00& 000     &00$\cdots$00& 000     &00$\cdots$0\alert{0}& \alert{1}10     &00$\cdots$00& 000    &00$\cdots$00& 0000000 &00$\cdots$00& 0011100\\
$\overline{x}_3$ &00$\cdots$00& 000     &00$\cdots$00& 000     &00$\cdots$00& 0\alert{11}     &\alert{0}0$\cdots$00& 000    &00$\cdots$00& 0000111 &00$\cdots$00& 0000000\\
$d_4$       &00$\cdots$00& 000     &00$\cdots$00& 000     &00$\cdots$00& 000     &1\alert{1$\cdots$1}1& 000    &00$\cdots$00& 0000000 &00$\cdots$00& 0000000\\
$x_4$       &00$\cdots$00& 000     &00$\cdots$00& 000     &00$\cdots$00& 000     &00$\cdots$0\alert{0}& \alert{1}10    &00$\cdots$00& 0000000 &00$\cdots$00& 0000000\\
$\overline{x}_4$ &00$\cdots$00& 000     &00$\cdots$00& 000     &00$\cdots$00& 000     &00$\cdots$00& 0\alert{11}    &\alert{0}0$\cdots$00& 0000000 &00$\cdots$00& 0000111\\
$d_5$       &00$\cdots$00& 000     &00$\cdots$00& 000     &00$\cdots$00& 000     &00$\cdots$00& 000    &1\alert{1$\cdots$1}1& 0000000 &00$\cdots$00& 0000000\\
$c_{1,1}$ & 00$\cdots$00    & 000     & 00$\cdots$00 & 000     & 00$\cdots$00 & 000     & 00$\cdots$00   & 000      & 00$\cdots$0\alert{0}  & \alert{111}1111   & 00$\cdots$00  & 0000000\\
$c_{1,2}$ & 00$\cdots$00    & 000     & 00$\cdots$00 & 000     & 00$\cdots$00 & 000     & 00$\cdots$00   & 000      & 00$\cdots$00  & 111\alert{11}11   & 00$\cdots$00  & 0000000\\
$c_{1,3}$ & 00$\cdots$00    & 000     & 00$\cdots$00 & 000     & 00$\cdots$00 & 000     & 00$\cdots$00   & 000      & 00$\cdots$00  & 11111\alert{11}   & \alert{0}0$\cdots$00  & 0000000\\
$d_6$       & 00$\cdots$00 & 000     &00$\cdots$00& 000     & 00$\cdots$00& 000     & 00$\cdots$00& 000     &00$\cdots$00& 0000000 &1\alert{1$\cdots$1}1& 0000000\\
$c_{2,1}$ & 00$\cdots$00    & 000      & 00$\cdots$00    & 000      & 00$\cdots$00    & 000      & 00$\cdots$00    & 000      & 00$\cdots$00  & 0000000   & 00$\cdots$0\alert{0}  & \alert{11}11111\\
$c_{2,2}$ & 00$\cdots$00    & 000      & 00$\cdots$00    & 000      & 00$\cdots$00    & 000      & 00$\cdots$00    & 000      & 00$\cdots$00  & 0000000   & 00$\cdots$00  & 11\alert{11}111\\
$c_{2,3}$ & 00$\cdots$00    & 000      & 00$\cdots$00    & 000      & 00$\cdots$00    & 000      & 00$\cdots$00    & 000      & 00$\cdots$00  & 0000000   & 00$\cdots$00  & 1111\alert{111}\\
\bottomrule
\end{tabular}
}
\end{table*}

We show that the reduced instance has an EF1 contiguous allocation if and only if the 3SAT formula is satisfiable.

Suppose that the given 3SAT formula $\varphi$ is satisfiable.
For a truth assignment that satisfies $\varphi$, let $I\subseteq [n]$ be the index set of variables assigned true.
In addition, for each $j\in[m]$, suppose that $\ell_{j,p_j}$ with $p_j\in\{1,2,3\}$ is true.
We construct an EF1 contiguous allocation from this assignment.
For each $i\in[n]$, the left item of $V_i$ is allocated to $x_i$ and the right item of $V_i$ is allocated to $\overline{x}_i$.
In addition, the middle item of $V_i$ is allocated to $x_i$ if $i\in I$ and to $\overline{x}_i$ otherwise.
For each $j\in[m]$, the items in $Q_j$ are allocated to $c_{j,1},c_{j,2},c_{j,3}$ where $c_{j,p_j}$ receives three items and each of the other two receives two items.
The divider items are allocated as follows:
\begin{itemize}
\item the leftmost item of $D_1$ is allocated to $d_1$,
\item the middle $r$ items of $D_k$ are allocated to $d_k$ for each $k\in[n+m]$,
\item the leftmost item of $D_{i+1}$ is allocated to $\overline{x}_i$ for each $i\in[n]$,
\item the rightmost item of $D_{i}$ is allocated to $x_i$ for each $i\in[n]$,
\item the leftmost item of $D_{n+j+1}$ is allocated to $c_{j,3}$ for each $j\in[m-1]$,
\item the rightmost item of $D_{n+j}$ is allocated to $c_{j,1}$ for each $j\in[m]$.
\end{itemize}
Then, it is not difficult to see that the allocation is EF1.

Conversely, suppose that an EF1 contiguous allocation $\bA$ exists for the reduced instance.
We show that the truth assignment such that the variables in $\{x\in N_x \mid v_{x}(A_x)=2\}$ are assigned to true is a satisfying assignment.
For each $k\in[n+m]$, agent $d_{k}$ must get at least one item in $D_k$, since otherwise, another agent gets at least two items of $D_k$ because $|D_k|>|N|$ and $d_{k}$ envies the agent by more than one item.
The rightmost two items of $D_1$ are not allocated to either $s_1$ or $s_2$. 
Otherwise, $d_1$ would receive at most one item from $D_1$ by the ordering of agents, while either $s_1$ or $s_2$ would receive at least $4$ items from $D_1$ due to $|D_1|=r+2=3n+4m+5\ge 9$, resulting in $d_1$ envying $s_1$ or $s_2$ by more than one item.
Thus, $v_{s_1}(A_{s_1})=v_{s_2}(A_{s_2})=0$.
This implies that the leftmost two items of $D_k$ with $k\in\{2,3,\dots,n+m\}$ are allocated to different agents.
Also, the rightmost two items of $D_k$ with $k\in\{1,2,\dots,n+m\}$ are allocated to different agents.
Hence, $A_{d_k}\subseteq D_k$ for each $k\in[n+m]$.
For each $i\in[n]$, the items in $V_i$ are allocated to $x_i$ or $\overline{x}_i$, and we have only two possibilities that 
$(v_{x_i}(A_{x_i}),v_{\overline{x}_i}(A_{\overline{x}_i}))$ is $(2,1)$ and $(1,2)$ since $\bA$ is EF1.
For each $j\in[m]$, the items in $Q_j$ are allocated to $c_{j,1},c_{j,2},c_{j,3}$, and we have only three possibilities that 
$(v_{c_{j,1}}(A_{c_{j,1}}),v_{c_{j,2}}(A_{c_{j,2}}),v_{c_{j,3}}(A_{c_{j,3}}))$ is $(3,2,2)$, $(2,3,2)$, and $(2,2,3)$ since $\bA$ is EF1.
Let $p_j\in\{1,2,3\}$ be the index $t$ such that $v_{c_{j,t}}(A_{c_{j,t}})=3$.
Then, by the EF1 condition, the agent $\ell_{j,p_j}\in N_x\cup N_{\overline{x}}$ must receive two desired items in $\bA$.
Therefore, the truth assignment such that the variables in $\{x\in N_x \mid v_{x}(A_x)=2\}$ are assigned to true is a satisfying assignment.
\end{proof}

\section{Flexible Order Setting} \label{sec:flexord}
In this section, we deal with the flexible-order setting of the problems, where no order of the agents is specified.

For the instance in Example~\ref{ex:fixed-order}, the allocation that agent $1$ receives $\{g_3,g_4\}$ and agent $2$ receives $\{g_1,g_2\}$ satisfies EF, U-max, E-max, MMS, PROP, and EQ. Note that this allocation is not order-consistent but contiguous.
It is known that an EF1 allocation~\cite{igarashi2023cut} and an MMS allocation~\cite{bouveret2017fair} always exist in the flexible order setting of CAP. 
Additionally, for each set of properties $X$ with $\emptyset\ne X\subseteq\{\textit{EF, PROP, EQ}\}$, deciding whether there exists a contiguous allocation satisfying all properties in $X$ is NP-complete~\cite{goldberg2020contiguous}.
Thus, we only focus on U-max and E-max problems.

\subsection{Approximation for U-max}
We begin with an observation for a special case where the valuations are binary additive and the maximum utilitarian social welfare equals the number of items $m$. We show below that, in this case, we can approximate the solution to a factor of $1/2$, which is better than the $1/8$-approximation algorithm~\cite{aumann2012computing} known for the general situation where the maximum utilitarian social welfare may be less than the number of items.

Suppose that there exists a contiguous allocation $\bA^*$ such that $\sum_{i\in N}v_i(A^*_i)=m$. Namely, $v_i(g)=1$ for any $i\in N$ and $g\in A^*_i$.
We present an algorithm that guarantees a utilitarian social welfare of at least $\lceil m/2 \rceil$. 
For each agent, a consecutive block of items that are valued by the agent is called a \emph{run}.

We denote $X$ as the set of allocated items and $I$ as the set of remaining agents.
At the beginning of our algorithm, $X$ is initialized as empty, and $I$ consists of all agents $N$.
In each iteration of the algorithm, it selects the longest run from among all the runs of agents in $I$ for items $M\setminus X$. 
Suppose that a longest run is $A_{i}\subseteq M\setminus X$ for agent $i\in I$.
Then, it allocates $A_{i}$ to $i$ and updates $X\leftarrow X\cup A_{i}$ and $I\leftarrow I\setminus\{i\}$.
The above process is repeated until $I$ becomes the empty set.
We denote the resulting partial allocation as $\bA=(A_1,\dots,A_n)$.
Finally, unallocated items are allocated to appropriate agents without violating contiguity. Let $\tilde{\bA}$ be the resulting extended allocation from $\bA$.
Our algorithm is formally described as Algorithm~\ref{alg:U-max Approx}.

\begin{algorithm}[H]
    \caption{U-max Approximation}\label{alg:U-max Approx}
    Let $X\ot \emptyset$ and $I\ot N$\;
    \While{$I\ne \emptyset$}{
        \tcc{Compute a longest run $R_i$ for each agent $i\in I$}
        \ForEach{$i\in I$}{
            $R_i\ot \emptyset$\;
            $j\ot 1$\; 
            \While{$j\le m$}{
                \lIf{$g_j\in X$ or $v_i(g_j)=0$}{\textbf{continue}}
                $k\ot j+1$\;
                \While{$k\le m$}{
                    \lIf{$g_k\not\in X$ and $v_i(g_k)=1$}{
                        $k\ot k+1$
                    }
                }
                \If{$k-j>|R_i|$}{
                    Update $R_i\ot \{g_j,g_{j+1},\dots,g_{k-1}\}$\;
                }
                Assign $j\ot k$\;
            }
        }
        Let $i^*\in\argmax_{i\in I}|R_{i}|$\;
        \lIf{$R_{i^*}=\emptyset$}{\textbf{break}}
        Update $A_{i^*}\ot R_{i^*}$, $X\ot X\cup R_{i^*}$, and $I\ot I\setminus\{i^*\}$ \tcc*{Allocate $R_i^*$ to $i^*$}
    }
    \ForEach{$i\in N\setminus I$}{
        Let $A_i=\{g_{k_i},g_{k_i+1},\dots,g_{\ell_i}\}$\;
    }
    \tcc{Extend the partial allocation $\bA$ to $\tilde{\bA}$}
    Let $\sigma\colon [N\setminus I]\to N\setminus I$ be the bijection such that $k_{\sigma(1)}<\dots<k_{\sigma(|N\setminus I|)}$ \tcc*{Allocation ordering}
    Let $k_{\sigma(|N\setminus I|+1)}\ot m+1$\;
    \For{$i\gets 1$ \KwTo $|N\setminus I|$}{
        $\tilde{A}_{\sigma(i)}\ot\{g_{k_{\sigma(i)}},g_{k_{\sigma(i)}+1},\dots,g_{k_{\sigma(i+1)}-1}\}$\;
    }
    $\tilde{A}_{\sigma(1)}\ot \tilde{A}_{\sigma(1)}\cup\{g_1,g_2,\dots,g_{k_{\sigma(1)}-1}\}$\;
    \ForEach{$i\in I$}{
        Let $\tilde{A}_i\ot \emptyset$\;
    }
    \Return $(\tilde{A}_1,\dots,\tilde{A}_n)$\;
\end{algorithm}

Now, we show that the approximation ratio of the above greedy algorithm is at least $1/2$. 
\begin{theorem}
For the flexible-order setting, the U-max problem of binary CAP admits a $1/2$-approximation algorithm if $\sum_{i\in N}v_i(A_i^*)=m$ for some contiguous allocation $\bA^*$.
\end{theorem}
\begin{proof}
We prove that the above algorithm outputs a contiguous allocation with the utilitarian social welfare of at least $\lceil m/2\rceil$ in polynomial time.
Let $X=\bigcup_{i\in N}A_i$.
The utilitarian social welfare of the allocation obtained by the algorithm is at least $\sum_{i\in N}v_i(\tilde{A}_i)\ge \sum_{i\in N}v_i(A_i)=\sum_{i\in N}|A_i|=|X|$.

For each agent $i\in N$, we will show that $|A_i|\ge |A^*_i\setminus X|$.
Until $A_i$ is allocated in the algorithm, items in $A_i^*$ that have not yet been allocated are contiguous because a run that separates the remaining items in $A_i^*$ is shorter than the longest run for $i$.
Thus, when $A_i$ is allocated in the algorithm, the number of remaining items in $A_i^*$ is at most $|A_i|$. This implies that $|A_i|\ge |A^*_i\setminus X|$.

By summing up the inequalities for all agents, we obtain $|X|=\sum_{i\in N}|A_i|\ge \sum_{i\in N}|A^*_i\setminus X|=\sum_{i\in N}|A^*_i|-|X|=m-|X|$.
Hence, the utilitarian social welfare of $\tilde{\bA}$ is at least $|X|\ge m/2$.
As the utilitarian social welfare is an integer, it is at least $\lceil m/2\rceil$.

Finally, we evaluate the computational complexity of the algorithm.
In each iteration, the longest run can be computed in $O(m|I|)=O(mn)$ time. The number of iterations is $n$.
The allocation $\tilde{\bA}$ can be obtained from $\bA$ in a linear time.
Therefore, the overall computational complexity is $O(mn^2)$, which is polynomial.
\end{proof}

Next, we show that there exists an $a$-approximation algorithm for the U-max problem when the valuations are $(a,b)$-sparse.
\begin{theorem}
For the flexible-order setting, the U-max problem of $(a,b)$-sparse CAP admits a $1/a$-approximation algorithm.
\end{theorem}
\begin{proof}
Construct a weighted complete bipartite graph between agents $N$ and items $M$ where the weight of $(i,g)\in N\times M$ is $v_i(g)$.
Then, compute a maximum weight matching $\mu\subseteq N\times M$ in the graph, which can be done in polynomial time.

We prove that allocating item $g$ to agent $i$ for each $(i,g)\in\mu$ is a $1/a$-approximate solution for the problem.
To see this, let $\bA^*$ be an optimal solution and let $\mu^*\subseteq N\times M$ be a matching where each agent $i\in N$ is matched to $g^*\in\argmax_{g\in A^*_i}v_i(g)$ if $A^*_i\ne\emptyset$ and unmatched otherwise.
Then, the optimum value is at most
\begin{align*}
\sum_{i\in N}v_i(A^*_i)
\ge a\cdot \sum_{i\in N}\sum_{g:\,(i,g)\in \mu^*}v_i(g)
\ge a\cdot \sum_{i\in N}\sum_{g:\,(i,g)\in \mu}v_i(g),
\end{align*}
since the valuations are $(a,b)$-sparse.
Hence, the allocation that assigns item $g$ to agent $i$ for each $(i,g)\in\mu$ is a $1/a$-approximate solution for the problem.
\end{proof}

\subsection{Parameterized algorithm for U-max} \label{sec:paramalg}
Here, we provide a parameterized algorithm for U-max.
Recall that $m$ is the number of items, and $n$ is the number of agents. 
As we have shown in Corollary~\ref{cor:Umax}, the problem has an FPT algorithm with parameter $n$. We will present an FPT algorithm with the parameter $m$.

Note that an FPT algorithm with the parameter of the maximum utilitarian social welfare $k$ can be obtained by the FPT algorithm with the parameter $m$, because $k\ge \sqrt{m}$ holds.
\begin{proposition}\label{prop:k>m}
$k\ge \sqrt{m}$.
\end{proposition}
\begin{proof}
Let $M_i=\{g\in M\mid v_i(g)\ge 1\}$ for each $i\in N$ and let $E=\{(i,g)\in N\times M\mid v_i(g)\ge 1\}$.
Suppose that $X^*=\{(i_1,g_{j_1}),(i_2,g_{j_2}),\dots,(i_\ell,g_{j_\ell})\}\subseteq N\times M$ is a maximum matching in the bipartite graph $(N,M;E)$.
Since we can construct a contiguous allocation with utilitarian social welfare at least $\ell$ by extending $X^*$, we have $\ell\le k$.
As $X^*$ is maximal, we have $v_i(g)=0$ for any $i\in N\setminus \{i_1,\dots,i_\ell\}$ and $g\in M\setminus\{g_{j_1},\dots,g_{j_\ell}\}$.
Hence, $M=\bigcup_{t=1}^\ell M_{i_t}$.
Moreover, we have $|M_i|\le k$ for all $i\in N$ since $v_i(M)=v_i(M_i)\le k$.
Thus, we obtain
\begin{align*}
m=|M|=\left|\bigcup_{t=1}^\ell M_{i_t}\right|
\le \sum_{t=1}^\ell|M_{i_t}|\le \ell k\le k^2,
\end{align*}
which means $k\ge \sqrt{m}$.
\end{proof}

Now, we provide an FPT algorithm with the parameter $m$.
Our algorithm enumerates all the possible contiguous partitions of $M$.
Note that there are only $2^{m-1}$ ways in total because there are two choices of whether or not to split between two consecutive items.
For each partition $(P_1,\dots,P_s)$, we construct a weighted complete bipartite graph between $N$ and $[s]$, where the weight of an edge $(i,t)\in N\times [s]$ is $v_i(P_t)$.
Then, compute the maximum weight matching of this graph and construct a contiguous allocation corresponding to the matching.
The algorithm outputs the optimal allocation obtained so far.
Since the algorithm explores all possibilities of partitions, it outputs an optimum allocation.
Our algorithm is formally described as Algorithm~\ref{alg:FPT-m}.

\begin{algorithm}[H]
    \caption{FPT with respect to the number of items}\label{alg:FPT-m}
    Let $\bA^*\ot (M,\emptyset,\dots,\emptyset)$ and $u^*\ot v_1(M)$\;
    \ForEach{contiguous partition $(P_1,\dots,P_s)$ of $M$}{
        Let $G$ be a complete bipartite graph between $N$ and $[s]$, and let $w\colon N\times[s]\to\mathbb{Z}_+$ be weights where $w(i,t)=v_i(P_t)$ for each $(i,t)\in N\times [s]$\;
        Compute a maximum weight matching $X^*$ on $(G,w)$ by the Hungarian method\;
        \If{$\sum_{(i,s)\in X^*}v_i(P_s)>u^*$ and $|X^*|=s$}{
            Update $\bA^*$ to be a contiguous allocation where $A^*_i=P_t$ for each $(i,t)\in X^*$ and 
            $u^*\ot \sum_{i\in N}v_i(A^*_i)$\;
        }
    }
    \Return $\bA^*$\;
\end{algorithm}

We obtain the following theorem.
\begin{theorem}
For the flexible-order setting, 
the U-max problem of CAP is FPT with respect to the number of agents $n$, the number of items $m$ and the optimum value $k$.
\end{theorem}
\begin{proof}
    There exists an FPT algorithm with respect to $n$ by Corollary~\ref{cor:Umax}. 
    We show that the time complexity of Algorithm~\ref{alg:FPT-m} is FPT with respect to $m$.
    There are $2^{m-1}$ possibilities of partitions.
    For each partition, we can compute the maximum weight matching in $O(m^2n)$ time by the Hungarian method.
    Hence, the total computational time is $O(m^22^m\cdot n)$, which is FPT.
    By Proposition~\ref{prop:k>m}, the total computational time with respect to $k$ is $O(k^4 2^{k^2}\cdot n)$, which is also FPT.
\end{proof}

We remark that our FPT algorithm can be extended to handle constraints other than the contiguous constraint.
In general, the number of ways to partition $m$ elements is the $m$th Bell number, which is bounded by $m^m$.
Hence, we can solve the U-max problem in $O(m^{m+2}n)$ time for general constraints.

\subsection{Approximation for E-max}
For the E-max problem, it is NP-hard to approximate within a factor of $1/2$ as we will see that deciding the optimal egalitarian social welfare is at least $2$ or at most $1$ is NP-hard.
We show that we can obtain a $1/a$-approximate solution in a polynomial-time if the instance is $(a,b)$-sparse.

\begin{theorem}
For the flexible-order setting, the E-max problem of $(a,b)$-sparse CAP admits a $1/a$-approximation algorithm.
\end{theorem}
\begin{proof}
If the number of agents is more than the number of items, then the optimum value is $0$. Thus, without loss of generality, we may assume that $|M|\ge |N|$.
Construct a weighted complete bipartite graph between agents $N$ and items $M$ where the weight of $(i,g)\in N\times M$ is $v_i(g)$.
Then, compute an agent-side perfect matching $\mu\subseteq N\times M$ in the graph that maximizes the minimum edge weight.
This matching $\mu$ can be found in polynomial time by iterating on value $k\in\{v_i(g)\mid (i,g)\in N\times M\}$ with the binary search and computing an agent-side perfect matching in the bipartite graph $(N,M;\{(i,g)\in N\times M\mid v_i(g)\ge k\})$.

We prove that allocating item $g$ to agent $i$ for each $(i,g)\in\mu$ is a $1/a$-approximate solution for the problem.
To see this, let $\bA^*$ be an optimal solution and let $\mu^*\subseteq N\times M$ be an agent-side perfect matching where each agent $i\in N$ is matched to $g^*\in\argmax_{g\in A^*_i}v_i(g)$.
Then, the optimum value is at most
\begin{align*}
\min_{i\in N}v_i(A^*_i)
=\min_{i\in N}\sum_{g:\,g\in A^*_i}v_i(g)
\ge a\cdot \min_{i\in N}\sum_{g:\,(i,g)\in \mu^*}v_i(g)
\ge a\cdot \min_{i\in N}\sum_{g:\,(i,g)\in \mu}v_i(g),
\end{align*}
since the valuations are $(a,b)$-sparse.
Hence, the allocation that assigns item $g$ to agent $i$ for each $(i,g)\in\mu$ is a $1/a$-approximate solution for the problem.
\end{proof}

\subsection{Parameterized algorithms for E-max}

Here, we provide parameterized algorithms for E-max.
As we have shown in \Cref{cor:Emax}, the problem has an FPT algorithm with respect to the number of agents $n$.
We will provide an FPT algorithm with respect to the number of items $m$.
Note that the E-max problem is unlikely to admit an FPT algorithm with respect to the optimum value because we will see that deciding the optimal egalitarian social welfare is at least $2$ or at most $1$ is NP-hard.

\begin{theorem}
For the flexible-order setting, the E-max problem of CAP is FPT with respect to the number of items and the number of agents.
\end{theorem}
\begin{proof}
There exists an FPT algorithm with respect to $n$ by \Cref{cor:Emax}.
It is sufficient to provide an FPT algorithm with respect to the number of items $m$.
If $m<n$, then the optimum value is $0$, and hence any allocation is optimum.
If $m\ge n$, then the FPT algorithm with respect to $n$ (\Cref{cor:Emax}) is also FPT with respect to $m$.
Thus, for both cases, we obtain an FPT algorithm.
\end{proof}

\subsection{NP-hardness of U-max} \label{sec:nputil}
In this subsection, we demonstrate that it is NP-complete to decide whether the optimal utilitarian social welfare is equal to the number of items for the binary CAP.

We reduce from \emph{2L-OCC-3SAT}. This version of 3SAT is where each literal, both positive and negative, occurs exactly twice in the clauses. Thus, each variable occurs four times in the clauses. 2L-OCC-3SAT and even its monotone version, are known to be NP-hard \cite{darmann2021simplified} 

We start with an instance of 2L-OCC-3SAT, $\varphi=C_1\wedge C_2\wedge\dots\wedge C_m$, where each $C_j$ is a clause of the form $C_j=\ell_{j,1}\vee \ell_{j,2}\vee \ell_{j,3}$
and each literal, both positive and negative, occurs exactly twice in the clauses.
In each clause $C_j$, every literal $\ell_{j,t}$ is one of the variables $x_1,\dots,x_n$ or their negations $\overline{x}_1,\dots,\overline{x}_n$. Note that $3m=4n$.

We construct an instance of binary CAP as follows.

We create \emph{variable agents} $N_x=\{x_1,\dots,x_n\}$, \emph{divider agents} $N_d=\{d_1,\dots,d_{n+1}\}$, and \emph{occurrence agents} $N_c=\{c_{j,t}\mid j\in[m],\,t\in[3]\}$.
The set of agents is $N=N_x\cup N_d\cup N_c$, which consists of $2n+3m+1=6n+1$ agents.

We use $5n+m+1=19n/3+1$ items, denoted as $M=\{g_1,g_2,\dots,g_{5n+m+1}\}$, which are arranged in sequence based on their indices.
For each $i\in[n]$, define 
\begin{align*}
V_i=\big\{g_{5i-4 },g_{5i-3},g_{5i-2},g_{5i-1},g_{5i}\big\}.
\end{align*}
We refer to items in this set as \emph{variable items} of $x_i$.
Within items in $V_i$, we call the first $2$ items the \emph{positive occurrence items}, the next item the \emph{divider item}, and the next $2$ items the \emph{negative occurrence items}. 
Note that there are $5n$ variable items by $\sum_{i=1}^n |V_i|=5n$.
We also call the item $g_{5n+1}$ divider item.
Consequently, the set of divider items is given by
\begin{align*}
D=\big\{g_{5i-2} \mid i\in[n]\big\}\cup\{g_{5n+1 }\},
\end{align*}
which comprises $n+1$ items.
For the remaining items, we call the item $g_{5n+1+j}$ \emph{clause item} of $C_j$ for each $j\in[m]$.

The valuations of agents are defined as follows.
\begin{itemize}
\item For each variable agent $x_i\in N_x$ with $i\in[n]$, we set $v_{x_i}(g)=1$ only for four items $g$ in $V_i\setminus D$. 

\item For each divider agent $d_i$ with $i\in[n]$, we set $v_{d_i}(g)$ to be $1$ only when $g$ is the $i$th divider item $g_{5i-2}$
and $0$ otherwise.
For the last divider agent $d_{n+1}$, we set $v_{d_{n+1}}(g)$ to be $1$ only when $g$ is the last divider item $g_{5n+1}$ and $0$ otherwise.

\item For each occurrence agent $c_{j,t}\in N_c$ with $j\in[m]$ and $t\in[3]$, we set $v_{c_{j,t}}(g)$ to be $1$ only for the following two items $g$. Suppose that $C_j$ is the $h$th clause containing literal $\ell_{j,t}$.
If $\ell_{j,t}$ is a positive variable $x_i$, then $v_{c_{j,t}}(g)=1$ when $g$ is $h$th positive occurrence item of $V_i$ or the clause item of $C_j$.
If $\ell_{j,t}$ is a negative variable $\overline{x}_i$, then $v_{c_{j,t}}(g)=1$ when $g$ is $h$th negative occurrence item of $V_i$ or the clause item of $C_j$.
\end{itemize}

From this definition, we can see that each agent approves at most four items. Additionally, each item is approved by at most three agents. Specifically, each variable item is approved by two agents, each divider item is approved by one agent, and each clause item is approved by three agents. Thus, this instance is $(4,3)$-sparse.

We demonstrate that the optimal utilitarian social welfare of the reduced instance is {$|M|$} if and only if the SAT formula $\varphi$ is satisfiable.
Intuitively, the truth assignment of $\varphi$ corresponds to whether the variable agent $x_i$ receives positive occurrence items or negative occurrence items of $V_i$.
For example, if $\varphi=C_1\wedge C_2\wedge C_3\wedge C_4$ where $C_1=x_1\vee \overline{x}_2\vee x_3$ , $C_2=\overline{x_1}\vee x_2\vee \overline{x}_3$, $C_3=x_1\vee x_2\vee x_3$ and $C_4=\overline{x_1}\vee \overline{x_2}\vee \overline{x}_3$, the valuations of the agents in the reduced instance is given as in \Cref{tab:flexible-Umax-hard}.

\begin{table}[htbp]
\centering
\caption{Reduced instance of the U-max problem of binary CAP from $(x_1\vee \overline{x}_2\vee x_3)\wedge(\overline{x_1}\vee x_2\vee \overline{x}_3)\wedge(x_1\vee x_2\vee x_3)\wedge(\overline{x_1}\vee \overline{x_2}\vee \overline{x}_3)$. The allocation represented by red color corresponds to a truth assignment of $(x_1,x_2,x_3)=(\true,\false,\false)$}.\label{tab:flexible-Umax-hard}
\tabcolsep = 1.2mm
\renewcommand{\arraystretch}{.5}
\begin{tabular}{c||c|c|d|c|c|c|c|d|c|c|c|c|d|c|c|d|c|c|c|c}
\toprule
& \multicolumn{5}{c|}{$V_1$} & \multicolumn{5}{c|}{$V_2$} & \multicolumn{5}{c|}{$V_3$} & & $C_1$ & $C_2$ & $C_3$ & $C_4$\\\hline
& $g_1$ & $g_2$ & \cellcolor{gray!20} $g_3$ & $g_4$ & $g_5$ & $g_6$ & $g_7$ &\cellcolor{gray!20} $g_8$ &  $g_9$ & $g_{10}$ & $g_{11}$ & $g_{12}$ & \cellcolor{gray!20}$g_{13}$ & $g_{14}$ & $g_{15}$ & \cellcolor{gray!20}$g_{16}$ & $g_{17}$ & $g_{18}$ & $g_{19}$ & $g_{20}$\\\midrule
$x_1$ & $\alert{1}$ & $\alert{1}$ & \cellcolor{gray!20} $0$ & $1$ & $1$ & $0$ & $0$ &\cellcolor{gray!20} $0$ &  $0$ & $0$ & $0$ & $0$ & \cellcolor{gray!20}$0$ & $0$ & $0$ & \cellcolor{gray!20}$0$ & $0$ & $0$ & $0$ & $0$\\
$x_2$ & $0$ & $0$ & \cellcolor{gray!20} $0$ & $0$ & $0$ & $1$ & $1$ &\cellcolor{gray!20} $0$ &  $\alert{1}$ & $\alert{1}$ & $0$ & $0$ & \cellcolor{gray!20}$0$ & $0$ & $0$ & \cellcolor{gray!20}$0$ & $0$ & $0$ & $0$ & $0$\\
$x_3$ & $0$ & $0$ & \cellcolor{gray!20} $0$ & $0$ & $0$ & $0$ & $0$ &\cellcolor{gray!20} $0$ &  $0$ & $0$ & $1$ & $1$ & \cellcolor{gray!20}$0$ & $\alert{1}$ & $\alert{1}$ & \cellcolor{gray!20}$0$ & $0$ & $0$ & $0$ & $0$\\\hline
$d_1$ & $0$ & $0$ & \cellcolor{gray!20} $\alert{1}$ & $0$ & $0$ & $0$ & $0$ &\cellcolor{gray!20} $0$ &  $0$ & $0$ & $0$ & $0$ & \cellcolor{gray!20}$0$ & $0$ & $0$ & \cellcolor{gray!20}$0$ & $0$ & $0$ & $0$ & $0$\\
$d_2$ & $0$ & $0$ & \cellcolor{gray!20} $0$ & $0$ & $0$ & $0$ & $0$ &\cellcolor{gray!20} $\alert{1}$ &  $0$ & $0$ & $0$ & $0$ & \cellcolor{gray!20}$0$ & $0$ & $0$ & \cellcolor{gray!20}$0$ & $0$ & $0$ & $0$ & $0$\\
$d_3$ & $0$ & $0$ & \cellcolor{gray!20} $0$ & $0$ & $0$ & $0$ & $0$ &\cellcolor{gray!20} $0$ &  $0$ & $0$ & $0$ & $0$ & \cellcolor{gray!20}$\alert{1}$ & $0$ & $0$ & \cellcolor{gray!20}$0$ & $0$ & $0$ & $0$ & $0$\\
$d_4$ & $0$ & $0$ & \cellcolor{gray!20} $0$ & $0$ & $0$ & $0$ & $0$ &\cellcolor{gray!20} $0$ &  $0$ & $0$ & $0$ & $0$ & \cellcolor{gray!20}$0$ & $0$ & $0$ & \cellcolor{gray!20}$\alert{1}$ & $0$ & $0$ & $0$ & $0$\\\hline
$c_{1,1}$ & $1$ & $0$ & \cellcolor{gray!20} $0$ & $0$ & $0$ & $0$ & $0$ &\cellcolor{gray!20} $0$ &  $0$ & $0$ & $0$ & $0$ & \cellcolor{gray!20}$0$ & $0$ & $0$ & \cellcolor{gray!20}$0$ & $\alert{1}$ & $0$ & $0$ & $0$\\
$c_{1,2}$ & $0$ & $0$ & \cellcolor{gray!20} $0$ & $0$ & $0$ & $0$ & $0$ &\cellcolor{gray!20} $0$ &  $1$ & $0$ & $0$ & $0$ & \cellcolor{gray!20}$0$ & $0$ & $0$ & \cellcolor{gray!20}$0$ & $1$ & $0$ & $0$ & $0$\\
$c_{1,3}$ & $0$ & $0$ & \cellcolor{gray!20} $0$ & $0$ & $0$ & $0$ & $0$ &\cellcolor{gray!20} $0$ &  $0$ & $0$ & $\alert{1}$ & $0$ & \cellcolor{gray!20}$0$ & $0$ & $0$ & \cellcolor{gray!20}$0$ & $1$ & $0$ & $0$ & $0$\\
$c_{2,1}$ & $0$ & $0$ & \cellcolor{gray!20} $0$ & $\alert{1}$ & $0$ & $0$ & $0$ &\cellcolor{gray!20} $0$ &  $0$ & $0$ & $0$ & $0$ & \cellcolor{gray!20}$0$ & $0$ & $0$ & \cellcolor{gray!20}$0$ & $0$ & $1$ & $0$ & $0$\\
$c_{2,2}$ & $0$ & $0$ & \cellcolor{gray!20} $0$ & $0$ & $0$ & $\alert{1}$ & $0$ &\cellcolor{gray!20} $0$ &  $0$ & $0$ & $0$ & $0$ & \cellcolor{gray!20}$0$ & $0$ & $0$ & \cellcolor{gray!20}$0$ & $0$ & $1$ & $0$ & $0$\\
$c_{2,3}$ & $0$ & $0$ & \cellcolor{gray!20} $0$ & $0$ & $0$ & $0$ & $0$ &\cellcolor{gray!20} $0$ &  $0$ & $0$ & $0$ & $0$ & \cellcolor{gray!20}$0$ & $1$ & $0$ & \cellcolor{gray!20}$0$ & $0$ & $\alert{1}$ & $0$ & $0$\\
$c_{3,1}$ & $0$ & $1$ & \cellcolor{gray!20} $0$ & $0$ & $0$ & $0$ & $0$ &\cellcolor{gray!20} $0$ &  $0$ & $0$ & $0$ & $0$ & \cellcolor{gray!20}$0$ & $0$ & $0$ & \cellcolor{gray!20}$0$ & $0$ & $0$ & $\alert{1}$ & $0$\\
$c_{3,2}$ & $0$ & $0$ & \cellcolor{gray!20} $0$ & $0$ & $0$ & $0$ & $\alert{1}$ &\cellcolor{gray!20} $0$ &  $0$ & $0$ & $0$ & $0$ & \cellcolor{gray!20}$0$ & $0$ & $0$ & \cellcolor{gray!20}$0$ & $0$ & $0$ & $1$ & $0$\\
$c_{3,3}$ & $0$ & $0$ & \cellcolor{gray!20} $0$ & $0$ & $0$ & $0$ & $0$ &\cellcolor{gray!20} $0$ &  $0$ & $0$ & $0$ & $\alert{1}$ & \cellcolor{gray!20}$0$ & $0$ & $0$ & \cellcolor{gray!20}$0$ & $0$ & $0$ & $1$ & $0$\\
$c_{4,1}$ & $0$ & $0$ & \cellcolor{gray!20} $0$ & $0$ & $\alert{1}$ & $0$ & $0$ &\cellcolor{gray!20} $0$ &  $0$ & $0$ & $0$ & $0$ & \cellcolor{gray!20}$0$ & $0$ & $0$ & \cellcolor{gray!20}$0$ & $0$ & $0$ & $0$ & $1$\\
$c_{4,2}$ & $0$ & $0$ & \cellcolor{gray!20} $0$ & $0$ & $0$ & $0$ & $0$ &\cellcolor{gray!20} $0$ &  $0$ & $1$ & $0$ & $0$ & \cellcolor{gray!20}$0$ & $0$ & $0$ & \cellcolor{gray!20}$0$ & $0$ & $0$ & $0$ & $\alert{1}$\\
$c_{4,3}$ & $0$ & $0$ & \cellcolor{gray!20} $0$ & $0$ & $0$ & $0$ & $0$ &\cellcolor{gray!20} $0$ &  $0$ & $0$ & $0$ & $0$ & \cellcolor{gray!20}$0$ & $0$ & $1$ & \cellcolor{gray!20}$0$ & $0$ & $0$ & $0$ & $1$\\
\bottomrule
\end{tabular}
\end{table}

Suppose that $\varphi$ is satisfiable.
By fixing a truth assignment that satisfies $\varphi$, we construct a contiguous allocation that achieves a utilitarian social welfare of $|M|=5n+m+1$ as follows.
\begin{itemize}
\item The $i$th divider item $g\in D$ is allocated to the $i$th divider agent for each $i\in[n+1]$.

\item For each $i\in[n]$ with which $x_i$ is assigned $\true$ in the truth assignment, the positive occurrence items in $V_i\setminus D$ are allocated to the variable agent $x_i\in N_x$.
Each negative occurrence item $g_{h}\in V_i\setminus D$ with $h\in \{5i-1,5i\}$ is allocated to the corresponding occurrence agent $c_{j,t}\in N_c$, who values it.

\item For each $i\in[n]$ with which $x_i$ is assigned $\false$ in the truth assignment, the negative occurrence items in $V_i\setminus D$ are allocated to the variable agent $x_i\in N_x$.
Each positive occurrence item $g_{h}\in V_i\setminus D$ with $h\in \{5i-4,5i-3\}$ is allocated to the corresponding occurrence agent $c_{j,t}\in N_c$, who values it.

\item Each clause item $g_{5n+1+j}$ of $C_j$ with $j\in[m]$ is allocated to an occurrence agent $c_{j,t}$ with $t\in[3]$ for which the literal $\ell_{j,t}$ is $\true$ in the truth assignment.
Since the truth assignment satisfies $\varphi$, such an occurrence agent must exist.
\end{itemize}
It is not difficult to see that every item is allocated to an agent who values it.
Furthermore, this allocation is contiguous because each divider or occurrence agent receives at most one item, and each variable agent $x_i$ receives one contiguous block of size $2$.
Hence, this allocation satisfies the desired conditions.

Conversely, suppose that there is a contiguous allocation $\bA^*$ such that $\sum_{i\in N}v_i(A^*_i)=|M|$.
Since each item is allocated to an agent who values it, each divider item is allocated to the corresponding divider agent. 
Additionally, each clause item $g_{5n+1+j}$ with $j\in[m]$ is allocated to one of the three corresponding occurrence agents $c_{j,1},c_{j,2},c_{j,3}$.
For each $j\in[m]$, let $c_{j,s_j}$ be the agent who receives $g_{5n+1+j}$, where $s_j\in [3]$.
We construct a truth assignment as follows:
\begin{itemize}
\item if $\ell_{j,s_j}$ is a positive variable $x_i$, then we assign $\true$ to $x_i$;
\item if $\ell_{j,s_j}$ is a negative variable $\overline{x}_i$, then we assign $\false$ to $x_i$.
\end{itemize}
If such a truth assignment exists, it clearly satisfies $\varphi$.
Now, what is left is to show that no variable appears as both positive and negative.

Suppose to the contrary that there exists an index $i\in[n]$ such that $x_i$ and $\overline{x}_i$ appear in $\{\ell_{1,s_1},\dots,\ell_{m,s_m}\}$.
Then, at least one positive occurrence item and one negative occurrence item in $V_i$ must be allocated to the variable agent $x_i\in N_x$.
Then, by the contiguity of $\bA^*$, the divider item $g_{ 5i-2}\in V_i$ is also allocated to the agent $x_i$.
This contradicts the assumption because the divider item $g_{5i-2}$ must be allocated to the divider agent $d_i$.

As the reduction takes only a polynomial time, we obtain the following theorem. 
\begin{theorem}\label{thm:Umax-hard}
    For the flexible-order setting, the U-max problem of binary (4,3)-sparse CAP is NP-hard.
    Moreover, it is NP-complete to determine whether the optimal utilitarian social welfare is equal to the number of items. 
\end{theorem}


By the proof of the above theorem, we can also obtain the following corollary.
\begin{corollary}
    For the flexible-order setting of binary CAP, 
    it is NP-complete to determine whether there exists an allocation such that 
    each agent receives at most two (contiguous) items and the utilitarian social welfare is equal to the number of items.
\end{corollary}

We remark that it is solvable in polynomial time to determine whether there exists an allocation such that 
each agent receives at most \emph{one} item and the utilitarian social welfare is equal to the number of items $m$. 
Indeed, this problem can be solved by determining the existence of item-side perfect matching in a bipartite graph $G=(N,M;E)$, where $E=\{(i,g)\in N\times M\mid v_i(g)=1\}$.

\subsection{NP-hardness of E-max} \label{sec:npegal}
Now, we show NP-hardness of the E-max problem of binary CAP.
We give a reduction from 3SAT.
Given a 3SAT formula $\varphi$, the set of agents is denoted as $N=N_x\cup N_d\cup N_c$, where we have $n$ variable agents $N_x=\{x_1,\dots,x_n\}$, 
$n+1$ divider agents $N_d=\{d_1,\dots,d_{n+1}\}$, and 
$3m$ occurrence agents $N_c=\bigcup_{j=1}^m\{c_{j,1},c_{j,2},c_{j,3}\}$, as before.

We use $6n+10m+2$ items, denoted by $M=\{g_1,g_2,\dots,g_{6n+10m+2}\}$. The items are arranged in sequence based on their indices.
The variable items of $x_i$ now consists of $2(o_i +\overline{o}_i+3)$ consecutive items
$$
V_i=\big\{g_{\sum_{p=1}^{i-1}2(o_p+\overline{o}_p+3)+h}\mid h\in[2(o_i+\overline{o}_i+3)]\big\}.
$$
Within the variable items in $V_i$, we call the first item a \emph{left assignment item}, 
the next $2o_i$ items the \emph{positive occurrence items}, 
the next item again a left assignment item, the next two items the \emph{divider items}, 
the next item a \emph{right assignment item}, 
the next {$2\overline{o}_i$} items the \emph{negative occurrence items}, 
and the next item is again a right assignment item. 
Note that there are $6n+6m$ variable items by $\sum_{p=1}^n2(o_p+\overline{o}_p+3)=6n+6m$.
The next two items $g_{6n+6m+1}$ and $g_{6n+6m+2}$ are divider items.
For each clause $C_j$ with $j\in[m]$, we have four corresponding clause items $\{g_{6n+6m+4j-1},g_{6n+6m+4j},g_{6n+6m+4j+1},g_{6n+6m+4j+2}\}$.

The valuations of agents are defined as follows.
\begin{itemize}
\item For each variable agent $x_i\in N_x$ with $i\in[n]$, we set $v_{x_i}(g)$ to be $1$ if $g$ is a left or right assignment item and $0$ otherwise.

\item For each divider agent $d_i$ with $i\in[n]$, we set $v_{d_i}(g)$ to be $1$ only when $g$ is a divider item in $V_i$ 
and $0$ otherwise.

\item For the last divider agent $d_{n+1}$, we set $v_{d_{n+1}}(g)$ to be $1$ only when $g\in\{g_{6n+6m+1},g_{6n+6m+2}\}$ and $0$ otherwise.

\item For each occurrence agent $c_{j,t}\in N_c$ with $j\in[m]$ and $t\in[3]$, we set $v_{c_{j,t}}(g)$ to be $1$ only for six items $g$. Suppose that $C_j$ is the $h$th clause containing literal $\ell_{j,t}$.
If $\ell_{j,t}$ is a positive variable $x_i$, then $v_{c_{j,t}}(g)=1$ when $g$ is either the $(2h-1)$st or the $2h$th positive occurrence item of $V_i$ or a clause item of $C_j$.
If $\ell_{j,t}$ is a negative variable $\overline{x}_i$, then $v_{c_{j,t}}(g)=1$ when $g$ is either the $(2h-1)$st or the $2h$th negative occurrence item of $V_i$ or a clause item of $C_j$.
\end{itemize}
From this definition, we can see that each agent approves at most six items, and each item is approved by at most three agents. Thus, this instance is $(6,3)$-sparse.

We demonstrate that the optimal egalitarian social welfare of the reduced instance is (at least) $2$ if and only if the 3SAT formula $\varphi$ is satisfiable.
Intuitively, the truth assignment of $\varphi$ corresponds to whether the variable agent $x_i$ receives right assignment items or left assignment items of $V_i$.
For example, if $\varphi=C_1\wedge C_2$ where $C_1=x_1\vee x_2\vee \overline{x}_3$ and $C_2=x_2\vee x_3\vee \overline{x}_4$, the valuations of the agents in the reduced instance is given as in \Cref{tab:flexible-Emax-hard}.

\begin{table}[htbp]
\centering
\caption{Reduced instance of the E-max problem of binary CAP from { $(x_1\vee x_2\vee\overline{x}_3)\wedge(x_2\vee x_3\vee \overline{x}_4)$}. The allocation represented by red color corresponds to a truth assignment of $(x_1,x_2,x_3,x_4)=(\true,\false,\false,\false)$.}\label{tab:flexible-Emax-hard}
\tabcolsep = 2mm
\renewcommand{\arraystretch}{.5}
\scalebox{1}{
\begin{tabular}{c||c|c|c|c|d|c|c}
\toprule
            & $V_1$ & $V_2$    & $V_3$    & $V_4$  &      & $C_1$  & $C_2$   \\\midrule
$x_1$       & $1001\divider{00}\alert{11}$& $000000\divider{00}00$ & $0000\divider{00}0000$ & $00\divider{00}0000$ & \divider{$00$} & $0000$ & $0000$  \\
$x_2$       & $0000\divider{00}00$& $\alert{100001}\divider{00}11$ & $0000\divider{00}0000$ & $00\divider{00}0000$ & \divider{$00$} & $0000$ & $0000$  \\
$x_3$       & $0000\divider{00}00$& $000000\divider{00}00$ & $\alert{1001}\divider{00}1001$ & $00\divider{00}0000$ & \divider{$00$} & $0000$ & $0000$  \\
$x_4$       & $0000\divider{00}00$& $000000\divider{00}00$ & $0000\divider{00}0000$ & $\alert{11}\divider{00}1001$ & \divider{$00$} & $0000$ & $0000$  \\\hline
$d_1$       & $0000\alert{\divider{11}}00$& $000000\divider{00}00$ & $0000\divider{00}0000$ & $00\divider{00}0000$ & \divider{$00$} & $0000$ & $0000$  \\
$d_2$       & $0000\divider{00}00$& $000000\alert{\divider{11}}00$ & $0000\divider{00}0000$ & $00\divider{00}0000$ & \divider{$00$} & $0000$ & $0000$  \\
$d_3$       & $0000\divider{00}00$& $000000\divider{00}00$ & $0000\alert{\divider{11}}0000$ & $00\divider{00}0000$ & \divider{$00$} & $0000$ & $0000$  \\
$d_4$       & $0000\divider{00}00$& $000000\divider{00}00$ & $0000\divider{00}0000$ & $00\alert{\divider{11}}0000$ & \divider{$00$} & $0000$ & $0000$  \\
$d_5$       & $0000\divider{00}00$& $000000\divider{00}00$ & $0000\divider{00}0000$ & $00\divider{00}0000$ & \alert{\divider{11}} & $0000$ & $0000$  \\\hline
$c_{1,1}$ & $0\alert{11}0\divider{00}00$& $000000\divider{00}00$ & $0000\divider{00}0000$ & $00\divider{00}0000$ & \divider{$00$} & $1111$ & $0000$  \\
$c_{1,2}$ & $0000\divider{00}00$& $011000\divider{00}00$ & $0000\divider{00}0000$ & $00\divider{00}0000$ & \divider{$00$} & $\alert{11}11$ & $0000$  \\
$c_{1,3}$ & $0000\divider{00}00$& $000000\divider{00}00$ & $0000\divider{00}0000$ & $00\divider{00}0000$ & \divider{$00$} & $11\alert{11}$ & $0000$  \\
$c_{2,1}$ & $0000\divider{00}00$& $000110\divider{00}00$ & $0000\divider{00}0000$ & $00\divider{00}0000$ & \divider{$00$} & $0000$ & $\alert{11}11$  \\
$c_{2,2}$ & $0000\divider{00}00$& $000000\divider{00}00$ & $0110\divider{00}0000$ & $00\divider{00}0000$ & \divider{$00$} & $0000$ & $11\alert{11}$  \\
$c_{2,3}$ & $0000\divider{00}00$& $000000\divider{00}00$ & $0000\divider{00}0000$ & $00\divider{00}0\alert{11}0$ & \divider{$00$} & $0000$ & $1111$  \\
\bottomrule
\end{tabular}
}
\end{table}

Suppose that $\varphi$ is satisfiable.
By fixing a truth assignment that satisfies $\varphi$, we construct an allocation that achieves an egalitarian social welfare of $2$ as follows.
\begin{itemize}
\item For each $i\in[n]$, $x_i\in N_x$ receives the two right and two left assignment items (including the items in between) if the variable $x_i$ is assigned $\true$ and $\false$ in the truth assignment, respectively.
\item For each $i\in[n]$, the divider agent $d_i$ receives two divider items in $V_i$.
The divider agent $d_{n+1}$ receives two divider items of $g_{6n+6m+1}$ and $g_{6n+6m+2}$.
\item For each $j\in[m]$, at least one literal in $\{\ell_{j,1},\ell_{j,2},\ell_{j,3}\}$ is $\true$ in the truth assignment. For an agent $c_{j,t}$ who corresponds to such literal, we allocate two variable items that $c_{j,t}$ desires. For the other two agents out of $\{c_{j,1},c_{j,2},c_{j,3}\}$, we allocate two clause items of $C_j$ to each.
\end{itemize}
Then, it is not difficult to see that this allocation is contiguous, and every agent gets the value of $2$.

Conversely, suppose that there exists a contiguous allocation $\bA^*$ such that $v_i(A^*_i)\ge 2$ for every $i\in N$.
Note that each divider agent must receive corresponding divider items.
Thus, each variable agent $x_i$ receives two right assignment items or two left assignment items of $V_i$.
Then, we construct a truth assignment as follows:
\begin{quote}
if $x_i\in N_x$ receives two right (resp.\ left) assignment items of $V_i$, then we assign $\true$ (resp.\ $\false$) to $x_i$.
\end{quote}

Suppose to the contrary, this assignment does not satisfy $\varphi$.
Then, there exists a clause $C_j$ that is not satisfied.
In this case, the clause agents $c_{j,1},c_{j,2},c_{j,3}$ cannot take desired variable items.
However, the remaining items that they desire are only four clause items of $C_j$.
Hence, the valuation for $\bA^*$ of at least one of these agents is at most $1$, which contradicts the assumption that $v_i(A^*_i)\ge 2$ for every $i\in N$.

Therefore, the following theorem holds.
\begin{theorem}\label{thm:Emax-hard}
For the flexible-order setting, the E-max problem of binary $(6,3)$-sparse CAP is NP-hard.
Specifically, it is NP-complete to determine whether the optimal egalitarian social welfare is at least $2$ or at most $1$.
\end{theorem}

It should be noted that whether the optimal egalitarian social welfare is at least $1$ can be efficiently checked by considering the existence of agent-side perfect matching in a bipartite graph $G=(N,M;E)$, where $E=\{(i,g)\in N\times M\mid v_i(g)=1\}$.


\section{Concluding remarks}
This study investigated the computational complexity of finding a fair or efficient contiguous allocation when items are arranged on a line, and each agent has a binary valuation for each item.
We considered two settings of fixed-order and flexible-order.
For the fixed-order setting, we provided polynomial-time algorithms for the problems of U-max, E-max, MMS, PROP, and EQ. Additionally, we proved the NP-hardness of checking the existence of an EF1 allocation.
For the flexible-order setting, we demonstrated that the problems of U-max and E-max are NP-hard.
Moreover, for the U-max problem, we provided a $1/2$-approximation algorithm for a special case and two FPT algorithms.

Finally, we discuss some possible future directions.
A straightforward future work is to construct a faster parameterized algorithm for maximizing utilitarian social welfare. 
There is also the scope of finding better approximation algorithms or better inapproximability of the U-max and E-max problems.
Specifically, it is open whether the $1/8$-approximation algorithm of Aumann et al.~\cite{aumann2012computing} for the U-max problem with additive valuations can be improved in the special case of binary additive.
Igarashi~\cite{igarashi2023cut} proved that an EF1 allocation always exists and posed a question regarding the computational complexity of finding it for the flexible-order setting. 
While we show that finding such an allocation is NP-hard for the fixed order case, the complexity of the problem for the general flexible order setting remains open.

\section*{Acknowledgement}
This work was partially supported by 
JST ERATO Grant Number JPMJER2301, 
JST PRESTO Grant Number JPMJPR2122, 
JSPS KAKENHI Grant Number JP20K19739,
Value Exchange Engineering, a joint research project between Mercari, Inc.\ and the RIISE, Sakura Science Exchange Program, and a MATRICS grant from SERB.

\bibliographystyle{abbrvnat} 
\bibliography{references}

\end{document}